%% file: JAPInfiniteSupportVer3.tex
\documentclass{aptpub}

\input{input.tex}

\authornames{Rahul Vaze, Srikanth Iyer} 
\shorttitle{Percolation on the Information-Theoretically Secure Signal to Interference Ratio Graph} 

\usepackage{graphicx}
\usepackage{amssymb}
\usepackage{amsfonts}
\usepackage{amsmath}
\usepackage{latexsym}
\usepackage{graphicx}
\usepackage{amsbsy}
\usepackage{epsfig}
\usepackage{color}
\usepackage{epstopdf}
\def\SINR{{\mathsf{SINR}}}
\begin{document}


\newcommand{\qed}{\hfill $\Box$}
\title{Percolation on the Information-Theoretically Secure Signal to Interference Ratio Graph}

\authorone[TIFR, Mumbai, India]{Rahul Vaze} 
\authortwo[IISc, Bangalore, India]{Srikanth Iyer}
\addressone{ School of Technology and Computer Science, Tata Institute of Fundamental Research, Homi Bhabha Road, Mumbai, India, 400005} 
\addressone{ Department of Mathematics, Indian Institute of Science, Bangalore, India, 560012}

{\sl Keywords:} Percolation, Information-Theoretic Security, SINR Graph, Wireless Communication.

%
%
\begin{center} {\bf Abstract} \end{center}
{We consider a continuum percolation model consisting of two types of nodes, namely legitimate and eavesdropper
nodes, distributed according to independent Poisson point processes (PPPs) in $\bbR ^2$ of intensities $\lambda$ and $\lambda_E$
respectively. A directed edge from one legitimate node $A$ to another legitimate node $B$ exists provided
the strength of the {\it signal} transmitted from node $A$ that is received at node $B$ is higher than that
received at any eavesdropper node. The strength of the received signal at a node from a legitimate node depends not
only on the distance between these nodes, but also on the location of the other legitimate nodes and an interference
suppression parameter $\gamma$. The graph is said to percolate when there exists an infinite connected
component. We show that for any finite intensity $\lambda_E$ of eavesdropper nodes,
there exists  a critical intensity $\lambda_c < \infty$ such that for all $\lambda > \lambda_c$ the graph
percolates for  sufficiently small values of the interference parameter. Furthermore, for the sub-critical regime, we show that there exists a $\lambda_0$
such that for all $\lambda < \lambda_0 \leq \lambda_c$ a 
suitable graph defined over eavesdropper node connections percolates that precludes percolation in the graphs formed by the legitimate nodes.}\\

%

\section{Introduction and main results}
Random geometric graphs have been used extensively to study various properties of wireless communication networks.
The nodes of the graph represent the communicating entities that are assumed to be distributed randomly in space, and the edges/connections between nodes reflect the realistic wireless communication links.
With the simplest connection model, two nodes are connected (or have an edge between them) provided they are within a specified cutoff distance from each other \cite{BookPenrose,BookRoy}.
Another connection model of interest is the protocol model \cite{Gupta2000}, that incorporates interference emanating from simultaneous transmission by multiple nodes, where two nodes are connected if there is no other node in a specified cut-off area (guard-zone) around the two nodes. Thus a smaller cutoff area results in greater spatial reuse, or more nodes being
able to communicate simultaneously. A non guard zone based connection model for wireless communication is the threshold model \cite{Gupta2000}, where two nodes are connected if the signal-to-noise-ratio (SINR) between them is more than a threshold.
The SINR measures the strength of the signal received from a particular node relative to those received from other nodes, and SINR between nodes $X_i$ and $X_j$ is defined as
\[ \SINR_{ij} = \frac{P \ell(X_i,X_j)}{N + \gamma \sum_{k \ne i,j} P \ell(X_k,X_j)}
\]
where $P$ is the transmitted power from each node, $\ell(X_i,X_j) \leq 1$ is the path-loss or attenuation factor, $N >0$ is the environment noise, and $\gamma \geq 0$ is the interference suppression parameter.

Existence of a path between two nodes in the graph implies the ability of those
nodes to communicate via a multi-hop path. Consequently, percolation in the graph corresponds to long range connectivity among large number of nodes that are part of the giant component. Assuming  that the nodes are distributed according to a Poisson point process in $\bbR^2$ of intensity $\lambda$,  existence of percolation in the graph with the SINR threshold connection model was shown in \cite{Dousse2005, Dousse2006} for all sufficiently small $\gamma > 0$.

Of recent interest is the problem of percolation in wireless
networks in the presence of eavesdroppers \cite{Haenggi2008a,  Pinto2010,Haenggi2011}. In these
models, referred to as the information theoretic secure models, a
legitimate node $i$ is connected (has an edge) to node $j$ provided node $j$ is
closer to node $i$ than its nearest eavesdropper. These are the
links over which secure communication can take place in the
presence of eavesdroppers of arbitrary capability.
Assuming that the legitimate and eavesdropper nodes are distributed
according to independent Poisson point processes in $\bbR^2$ of
intensities $\lambda$, and $\lambda_E$, respectively, existence of
phase transition of percolation in these graphs was established in \cite{Haenggi2008a,  Pinto2010,Haenggi2011}. Using a
branching process argument \cite{Haenggi2008a, Haenggi2011} show that if the ratio $\lambda/\lambda_E < 1$, then almost
surely, no unbounded connected component exits.

The above secrecy graph model \cite{Haenggi2008a,  Pinto2010,Haenggi2011} assumes that the signals transmitted from different
legitimate nodes do not interfere with each other. In reality, that is difficult to
incorporate, since there are large number of legitimate nodes, and all cannot transmit on orthogonal frequency or time slots.
To generalize the secrecy graph model, we extend the notion of the secrecy graph using the SINR or threshold model, where two legitimate nodes are connected if the SINR between them is more than the SINR at any other eavesdropper node.
We derive
two results on the percolation properties for this new SINR based secrecy graph model. The first result is
similar in spirit to the one derived by \cite{Dousse2006}
. It states that for any given intensity $\lambda_E$ of the
eavesdropper nodes, the secrecy graph percolates for sufficiently
large intensity $\lambda$ of legitimate nodes and all sufficiently
small interference suppression parameter $\gamma$. The second
result is that for a given $\lambda_E$ and $\gamma > 0$,
if the density of legitimate nodes is below a threshold,
then the graph does not percolate. To prove the second result, we use a novel technique of defining a   
suitable graph over eavesdropper node connections, where percolation in the eavesdropper nodes' graph precludes  percolation in the graphs formed by the legitimate nodes.
 To complete the result we show that for any given $\lambda_E$ and $\gamma > 0$,
if the density of legitimate nodes is below a threshold, then the defined eavesdropper nodes' graph percolates. \\

Before we proceed to describe the model in detail and state the
main results, we need some notation.

{\bf Notation:} The cardinality of set $S$ is denoted by $|S|$.
The complement of set $S$ is denoted by $\bar S$. A ball of radius
$r$ centered at $x$ is denoted by $B(x,r)$. The boundary of a set
$G \subset \bbR^2$ is denoted by $\delta G$. For a set $A \subset
\bbR^2$, $a+A$ denotes a translation of $A$ with $a \in \bbR^2$ as
the center. The Lebesgue measure of a set $ A \subset \bbR^2$ is
denoted as $\nu(A)$.

\subsection{System Model} We now describe the secure SINR graph (SSG), which
generalizes the secrecy graph considered in \cite{Haenggi2008a,  Pinto2010,Haenggi2011},
by allowing all legitimate nodes to transmit at the same
time/frequency and interfere with each other's communication. Let
$\Phi$ be the set of legitimate nodes, and $\Phi_E$ be the set of
eavesdropper nodes. We assume that the points in $\Phi$ and
$\Phi_E$ are distributed according to independent PPPs with intensities $\lambda$ and $\lambda_E$, respectively.
Let $x_i,x_j\in\Phi$, and $e\in \Phi_E$. Without loss of
generality, we assume an average power constraint of unity $(P=1)$ at each
node in $\Phi$, and noise variance $N=1$. Let $0<\gamma \le 1 $ be the processing gain of
the system (interference suppression parameter), which depends on
the orthogonality between codes used by different legitimate nodes
during simultaneous transmissions. Then the SINR between $x_i$ and
$x_j$ is
$$\SINR_{ij} \bydef \frac{\ell(x_i, x_j)}{\gamma \sum_{k\in \Phi
k\ne i}\ell(x_k, x_j) + 1},$$
and between $x_i$ and $e$ is
$$\SINR_{ie} \bydef \frac{\ell(x_i, e)}{\sum_{k\in \Phi k\ne
i}\ell(x_k, x_j) + 1}.$$
Note that the parameter $\gamma$ is absent in the second SINR
formula. This is due to the fact that the code used by the
legitimate nodes is not known to the eavesdroppers, hence no
processing gain can be obtained at any of the eavesdroppers. Then
the maximum rate of reliable communication between $x_i$ and $x_j$
such that an eavesdropper $e$ gets no knowledge is
\cite{Wyner1975}
$$R^{\SINR}_{ij}(e) \bydef \left[ \log_2\left(1+ \SINR_{ij}\right) - \log_2 \left(1+ \SINR_{ie}\right)\right]^+,$$
and the maximum rate of communication between $x_i$ and $x_j$ that is secured from  all the eavesdropper nodes of $\Phi_E$,
$$R^{\SINR}_{ij} \bydef \min_{e\in \Phi_E} R_{ij}(e).$$


\begin{defn} SINR  Secrecy graph (SSG) is a directed graph $SSG(\theta) \bydef \{\Phi, {\cal E}\}$, with vertex set $\Phi$, and edge
set ${\cal E} \bydef \{(x_i, x_j): R^{\SINR}_{ij}> \theta\}$, where $\theta$ is the minimum rate of secure communication required between any two nodes of $\Phi$.
\end{defn}
We will assume $\theta=0$ for the rest of the paper, and represent
$SSG(0)$ as $SSG$. The results can be generalized easily for
$\theta >0$. With $\theta=0$, $SSG \bydef \{\Phi, {\cal E}\}$,
with  edge set ${\cal E} \bydef \{(x_i, x_j): \SINR_{ij}>
\SINR_{ie}, \ \forall \ e\in \Phi_E \}$.

\begin{defn} We define that a node $x_i$ can {\it connect}  to  $x_j$ (or there is a link/connection between them) if  $(x_i, x_j) \in SSG$.\end{defn}

\begin{defn} We define that there is a {\it path} from node $x_i \in \Phi$ to $x_j \in \Phi$ if there is a connected path from $x_i$ to $x_j$ in the $SSG$.  A path between $x_i$ and $x_j$ on $SSG$ is represented as $x_i \rightarrow x_j$.
\end{defn}

\begin{defn}
The connected component of any node $x_j \in \Phi$, is defined as $C_{x_j} \bydef   \{x_k \in \Phi, x_j\rightarrow x_k\}$, with cardinality  $|C_{x_j}|$.
\end{defn}

\begin{rem} Note that because of stationarity of the PPP, the distribution of $|C_{x_j}|$ does not depend on $j$,
and hence without loss of generality from here on we consider node
$x_1$ for the purposes of defining connected components. Further
we assume without loss of generality that $x_1$ is at the origin.
\end{rem}


In this paper we are interested in studying the percolation
properties of the $SSG$. In particular, we are interested in
finding the minimum value $\lambda_c$ of $\lambda$ for which the
probability of having an unbounded connected component in $SSG$ is
greater than zero, as a function of  $\lambda_E$, i.e. $\lambda_c
\bydef \inf \{{\lambda}:P(|{\cal C}_{x_1}| = \infty)>0\}$. The
event $\{|{\cal C}_{x_1}| = \infty\}$ is also referred to as {\it
percolation} on $SSG$, and we say that percolation happens if
$P(\{|{\cal C}_{x_1}| = \infty\})>0$, and does not happen if
$P(\{|{\cal C}_{x_1}| = \infty\})=0$.

\begin{rem} Assuming that all legitimate nodes can transmit in orthogonal
time/frequency slots, secrecy graph $SG$ was introduced in
\cite{Haenggi2008a}, where two legitimate nodes are connected if
the received signal power between them is more than the received
signal power  at the nearest eavesdropper, i.e. $SG \bydef \{\Phi,
{\cal E}\}$, with vertex set $\Phi$, and edge set ${\cal E} \bydef
\{(x_i, x_j): \ell(x_i,x_j)> \ell(x_i, e), \ \forall \ e \in
\Phi_E\}$. Percolation properties of $SG$ were studied in
\cite{Haenggi2011, Pinto2010}, where in \cite{Haenggi2011} it was
shown that if $\lambda < \lambda_E$, then there is no percolation,
while \cite{Pinto2010} showed the existence of $\lambda$ for any
fixed $\lambda_E$ for which the $SG$ percolates. The graph structure
of $SSG$ is more complicated compared to $SG$ because of the presence
of interference power terms corresponding to simultaneously
transmitting legitimate nodes, and hence the results of
\cite{Haenggi2011, Pinto2010} do not apply for $SSG$. For example,
consider the case of $\gamma=0$, where it is possible that two
legitimate nodes $x_i$ and $x_j$, with $d_{ij} > \min_{e\in
\Phi_E} d_{ie}$ can connect to each other in the $SSG$, however,
$x_i$ and $x_j$ cannot connect to each other in the $SG$ since
$d_{ij} > \min_{e\in \Phi_E}d_{ie}$. Similarly, if $x_j$ is closer
to $x_i$ than any other eavesdropper node, then $x_i$ is connected
to $x_j$ in $SG$, however, that may not be the case in $SSG$.
\end{rem}

\begin{rem} Without the presence of eavesdropper nodes, percolation
on the SINR graph, where the vertex set is $\Phi$, and edge set
${\cal E} \bydef \{(x_i, x_j): {\SINR}_{ij}\ge \beta, x_i,x_j\in
\Phi\}$ for some fixed threshold $\beta$, has been studied in
\cite{Dousse2005, Dousse2006, VazeSINR2011}. The results of
\cite{Dousse2005, Dousse2006, VazeSINR2011}, however, do not apply
for the $SSG$, since for $SSG$, $\beta=\SINR_{ie}$ is a random
variable that depends on both $\Phi$ and $\Phi_E$.
\end{rem}

\begin{rem} Note that we have defined $SSG$ to be  a directed graph,
and the connected component of $x_1$ is its out-component, i.e. the
set of nodes with which $x_1$ can communicate secretly.
Since $x_i \rightarrow x_j, \ x_i,x_j\in \Phi$, does not imply
$x_j \rightarrow x_i \ x_i,x_j\in \Phi$,  one can similarly define
in-component $C_{x_j}^{in} \bydef   \{x_k \in \Phi, x_k\rightarrow
x_j\}$,  bi-directional component $C_{x_j}^{bd} \bydef   \{x_k \in
\Phi, x_j\rightarrow x_k\ \text{and} \ x_k\rightarrow x_j\}$, and
either one-directional component $C_{x_j}^{ed} \bydef   \{x_k \in
\Phi, x_j\rightarrow x_k \ \text{or} \ x_k\rightarrow x_j \}$.
Percolation on $C_{x_j}^{in}$, $C_{x_j}^{bd}$ and $C_{x_j}^{ed}$
is in principle similar to the percolation on out-component, but
are not considered in this paper.
\end{rem}

\subsection{Main Results}

\begin{thm}\label{thm:sup} For the signal attenuation function $\ell(x)$,
such that $\int x\ell(x) dx < \infty$, for any $\lambda_E$,
there exists $\lambda' < \infty$ and a function
$\gamma'(\lambda,\lambda_E) >0$, such that  $P(|C_{x_1}| = \infty)
> 0$ in the $SSG$ for $\lambda > \lambda'$, and $\gamma <
\gamma'(\lambda,\lambda_E)$.
\end{thm}

We show that for small enough $\gamma$, there exists a large enough $\lambda$ for which the $SSG$ percolates with positive probability for any value of $\lambda_E$. This result is similar in spirit to \cite{Dousse2005, Dousse2006}, where percolation is shown to happen in the SINR graph, where two nodes are connected if the SINR between them is more than a fixed threshold $\beta$, (without the secrecy constraint due to eavesdroppers) for small enough
$\gamma$ with finite and unbounded support signal attenuation function, respectively.
The major difference between the $SSG$ and SINR graph, is that with $SSG$ the threshold for connection between two nodes (maximum of SINRs received at all eavesdroppers) is a random variable that depends on both the legitimate and eavesdropper density, in contrast the threshold  in the SINR graph is a fixed constant.

To prove the result, we consider percolation on another graph $SSG$$^e$ that is a subset of $SSG$. $SSG$$^e$ is obtained from $SSG$ by replacing the SINR at each eavesdropper node in $SSG$ definition by $\SINR_{ie} = \ell(d_{ie})$, i.e. the SINR at each eavesdropper node is replaced by just the signal power received at the eavesdropper node and making the interference power terms equal to zero.
Considering this subset $SSG$$^e$ simplifies the percolation analysis significantly.
Then to show the percolation on the subset $SSG^e$, we map the continuum percolation of $SSG$ to an appropriate bond percolation on the square lattice,
similar to \cite{Dousse2006}.

For the converse, we have the following Theorem on the lower bound for the critical density $\lambda_c$.

\begin{thm} \label{thm:inf} For every $\lambda_E > 0$ and $\gamma
\in (0,1)$, there exists a $\lambda_c =
\lambda_c(\lambda_E,\gamma) > 0$ such that for all $\lambda <
\lambda_c$, $P(|C_{x_1}| = \infty) = 0$ in the $SSG$.
\end{thm}

We show that for any $\gamma>0$, there exists small enough $\lambda$ for which the $SSG$ does not percolate for any  value of $\lambda_E$.
In prior work, on secrecy graph with no interference among simultaneously transmitting legitimate nodes, a stronger result was proved that if $\lambda<\lambda_E$ then the secrecy graph does not percolate  \cite{Haenggi2011} using branching process argument on the out-degree distribution. We are only able to show an existential result for the $SSG$, since finding the out-degree distribution of any node in the $SSG$ is quite challenging and $SSG$ is not amenable to analysis similar to \cite{Haenggi2011}.

The proof idea is to define an appropriate eavesdropper node graph such that if an edge exists between two eavesdropper nodes then there exists no edge of $SSG$ that crosses that edge in $\bbR^2$. Note that for $SSG$ to percolate, there should be left to right crossing and top to bottom crossing of any square box of large size in $\bbR^2$ by connected edges of $SSG$.
However, if the eavesdropper node graph percolates, then there cannot be  left to right crossing and top to bottom crossing of any square box of large size in $\bbR^2$ by connected edges of $SSG$, and consequently $SSG$ cannot percolate if the eavesdropper node graph percolates.
Then we derive conditions for percolation on the defined eavesdropper node graph to find conditions when the $SSG$ does not percolate.

\section{Proof of Theorem~\ref{thm:sup}}
\label{sec:plsupcfading} In this section we  are interested in the
super-critical regime and want to find an upper bound  on
$\lambda$ such that $P(|C_{x_1}|=\infty)>0$ for a fixed
$\lambda_E$.  Towards that end, we will tie up the percolation on
$SSG$ to a bond percolation on square lattice, and show that bond
percolation on the square lattice implies percolation in the
$SSG$.

For the super-critical regime, we consider the enhanced graph $SSG^e$, where
$SSG^e \bydef \{\Phi, {\cal E}^e\}$, with  edge
set ${\cal E}^e \bydef \{(x_i, x_j): \SINR_{ij}> \ell(d_{ie}), \ \forall \ e\in \Phi_E \}$. For defining $SSG^e$, we have considered the interference power at the eavesdropper nodes to be zero.
Clearly, $SSG^e \subseteq SSG$, and hence if $SSG^e$ percolates, then so does $SSG$.

We tile $\bbR^2$ into a square lattice $\bS$ with side $s$. Let
$\bS' = \bS+(\frac{s}{2}, \frac{s}{2})$ be the dual lattice of
$\bS$ obtained by translating each edge of $\bS$ by $(\frac{s}{2},
\frac{s}{2})$. For any edge $\ba$ of $\bS$, let $S_1(\ba)$ and
$S_2(\ba)$ be the two adjacent squares to $\ba$. See Fig.
\ref{fig:supcrit} for a pictorial description. Let
$\{a_i\}_{i=1}^{4}$ denote the four  vertices of the rectangle
$S_1(\ba) \cup S_2(\ba)$. Let $Y(\ba)$ be the smallest square
containing $\cup_{i=1}^4 B(a_i, t)$, where $t$ is such that
$\ell(t)< \frac{\ell(\sqrt{5}s)}{2}$.

\begin{defn}\label{def:open} Any edge $\ba$ of $\bS$ is defined to be open if
\begin{enumerate}
\item there is at least one node of $\Phi$ in both the adjacent squares $S_1(\ba)$ and $S_2(\ba)$,
\item there are no eavesdropper nodes in $Y(\ba)$,
\item and for any legitimate node $x_i \in \Phi \cap \left(S_1(\ba)\cup S_2(\ba)\right)$, the interference received at any legitimate node $x_j \in \Phi \cap \left(S_1(\ba)\cup S_2(\ba)\right)$, $I_j^i \bydef \sum_{k\in \Phi, k\ne i }\ell(x_k,x_j) \le  \frac{1}{\gamma}$.\end{enumerate}
\end{defn}
An open edge is pictorially described in Fig. \ref{fig:supcrit} by edge $a$, where the  black dots represent a legitimate node while a cross
represents an eavesdropper node.

The next Lemma allows us to tie up the continuum percolation on $SSG$ to the bond percolation on the square lattice,
where we show that if an edge $\ba$ is open, then all legitimate nodes lying in $S_1(\ba) \cup S_2(\ba)$ can connect
to each other.
\begin{lemma}\label{lem:open} If an edge $\ba$ of $\bS$ is open, then any node $x_i \in \Phi \cap \left(S_1(\ba)\cup S_2(\ba)\right) $ can connect to any node $x_j \in \Phi \cap \left(S_1(\ba)\cup S_2(\ba)\right)$ in $SSG^e$.
\end{lemma}
\begin{proof} For any $x_i, x_j \in \Phi\cap\left(S_1(\ba)\cup S_2(\ba)\right)$, $\SINR_{ij} \ge \frac{\ell(\sqrt{5}s)}{2}$, since $I_j^i \bydef \sum_{k\in \Phi, k\ne i }\ell(x_k,x_j) \le  \frac{1}{\gamma}$ for
each $x_i, x_j \in \Phi \cap \left(S_1(\ba)\cup S_2(\ba)\right)$. Moreover, since there are no eavesdropper nodes in $Y(\ba)$, the minimum distance between any eavesdropper node from any legitimate node in $\Phi\cap\left(S_1(\ba)\cup S_2(\ba)\right)$ is at least $t$. Since $t$ is such that $\ell(t)< \frac{\ell(\sqrt{5}s)}{2}$, clearly, $x_i, x_j \in \Phi\cap\left(S_1(\ba)\cup S_2(\ba)\right)$ are connected in $SSG^e$.
\end{proof}

\begin{defn} An open component of $\bS$ is the sequence of connected open edges of $\bS$.
\end{defn}

\begin{defn} A circuit in $\bS$ or $\bS'$ is a connected path of $\bS$ or $\bS'$ which starts and ends at the same point.
A circuit in  $\bS$ or $\bS'$ is defined to be open/closed if all the edges on the circuit are open/closed in $\bS$ or $\bS'$.
\end{defn}

Some important properties of $\bS$ and $\bS'$ which are immediate are as follows.
\begin{lemma}\label{lem:int1} If the cardinality of the open component of $\bS$ containing the origin is infinite, then $|C_{x_1}|= \infty$.
\end{lemma}
\begin{proof} Follows from Lemma \ref{lem:open}.
\end{proof}


\begin{figure} [ht]
\centering
\scalebox{.5}{\input{JAPSupCrit.pstex_t}}
\vspace{0.1in}
\caption{Open edge definition on a square lattice for super-critical regime.}
\label{fig:supcrit}
\end{figure}
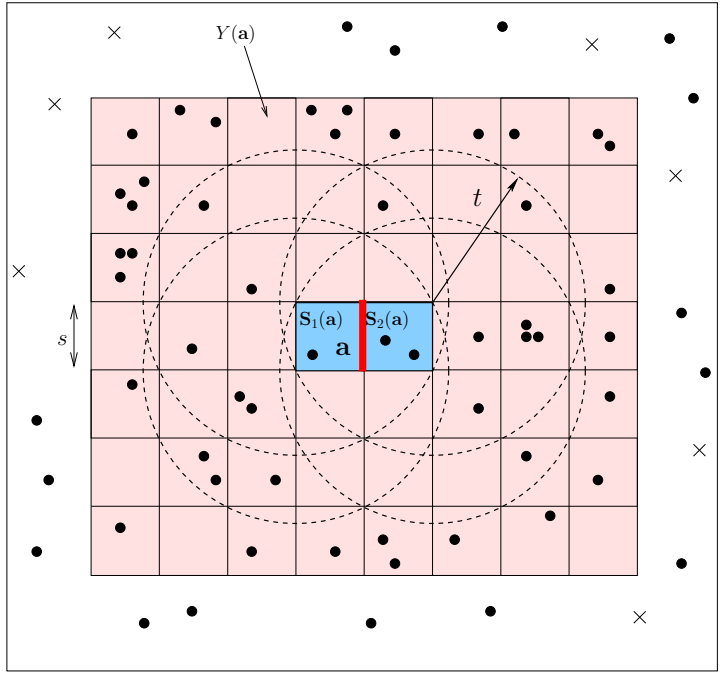

\begin{lemma}\label{lem:int2}\cite{Grimmett1980} The open component of $\bS$  containing the origin is  finite if and only if  there is a closed circuit in $\bS'$  surrounding the origin.
\end{lemma}

Hence, if we can show that the probability that there exists a closed circuit in $\bS'$ surrounding the origin is less than one, then it follows that an unbounded connected component exists in $\bS$ with non-zero probability. Moreover, having an unbounded connected component in the square lattice $\bS$ implies that there is an unbounded connected component in $SSG$ from Lemma \ref{lem:open}. Next, we find a bound on $\lambda$ as a function of $\lambda_E$ such that probability of having a closed circuit in $\bS'$  surrounding the origin is less than one. This is a standard approach used for establishing the existence of percolation in discrete graphs.

For an edge $\ba$, let $A(\ba)=1$ if $\Phi \cap S_i(\ba) \ne
\emptyset$, $i=1,2$,
and zero otherwise. Similarly, let $B(\ba)=1$ ($= 0$) if $I_j^i \bydef \sum_{k\in \Phi, k\ne i }\ell(x_k,x_j) \le  \frac{1}{\gamma}$
for $x_i, x_j \in \Phi \cap \left(S_1(\ba)\cup S_2(\ba)\right)$ (otherwise), and $C(\ba)=1$ ($= 0$) if there are no eavesdropper nodes in $Y(\ba)$ (otherwise).
Then by definition, the edge $\ba$ is open if $D(\ba) = A(\ba)B(\ba)C(\ba)=1$.

Now we want to bound the probability of having a closed circuit
surrounding the origin in $\bS$. Towards that end, we will first
bound the probability of a closed circuit of length $n$, i.e.
$P(D(\ba_1)= 0, D(\ba_2)=0, \ldots, D(\ba_n)=0), \ \forall \ n \in
\bbN$ considering $n$ distinct edges. Let $p_A \bydef
P(A(\ba_i)=0)$ for any $i$. Since $\Phi$ is a PPP with density
$\lambda$, $p_A = 1-(1-e^{-\lambda s^2 })^2$. Then we have the
following intermediate results to upper bound $ P(D(\ba_1)=0,
D(\ba_2)=0, \ldots, D(\ba_n)=0)$.
\begin{lemma}\label{lem:a} $P(A(\ba_1)=0, A(\ba_2)=0, \ldots, A(\ba_n)=0) \le   p_1^{n}$, where $p_1 \bydef p_{A}^{1/4}$.
\end{lemma}
\begin{proof} Follows from the fact that in any sequence of  $n$ edges of $\bS $ there are at least $n/4$ edges such that their adjacent squares
$S_1(\ba_e)\cup S_2(\ba_e)$ do not overlap. Therefore $P(A(\ba_1)=0, A(\ba_2)=0, \ldots, A(\ba_n)=0) \le P(\cap_{e\in O}A(\ba_e)=0)$, where $O$ is the
set of edges for which their adjacent squares $S_1(\ba_e)\cup S_2(\ba_e)$ have no overlap, and  $|O|=n/4$. Since $S_1(\ba_e)\cup S_2(\ba_e), \ e \in O$
have no overlap, and events $A(\ba_e)=0$ are independent for $\ba_e\in O$,  the result follows.
\end{proof}
\begin{lemma}\label{lem:b} \cite[Proposition 2]{Dousse2006} For $\int_{0}^{\infty} x\ell(x) dx < \infty$, $P(B(\ba_1)=0, B(\ba_2)=0, \ldots, B(\ba_n)=0) \le  p_2^{n}$,
where $p_2 \bydef e^{\left(\frac{2\lambda}{K}\int \ell(x) dx - \frac{1}{\gamma K} \right)}$, and $K$ is a constant.
\end{lemma}
\begin{lemma}\label{lem:c} $P(C(\ba_1)=0, C(\ba_2)=0, \ldots, C(\ba_n)=0) \le   p_3^{n}$, for some $p_3$ independent of $n$.
\end{lemma}
\begin{proof}
By definition, events $C(\ba_i)$ and $C(\ba_j)$ are independent if $Y(\ba_i) \cap Y(\ba_j) = \phi$. Consider a circuit ${\cal P}_n$ in $\bS$ of length $n$, with a subset ${\cal P}_n^s \subset {\cal P}_n$, where ${\cal P}_n^s = \{\ba_i\}_{i\in{\cal I}}$, where for any $n, m \in {\cal I},
Y(\ba_n) \cap Y(\ba_m) = \phi$. Since $Y(\ba)$ occupies at most $\left(L+\left \lceil \frac{2t}{s}\right\rceil\right) \times \left(L+1+\left \lceil \frac{2t}{s}\right\rceil\right)$ squares of lattice $\bS$,
where $L= 2 \left \lceil \sqrt{5} \right\rceil$, it follows that
$|{\cal I}| \ge \frac{n}{\psi}$, where $\psi = 8\left(L+\left \lceil \frac{2t}{s}\right\rceil\right)^2-1$.
Thus, $P(C(\ba_1)=0, C(\ba_2)=0, \ldots, C(\ba_n)=0) \le p_3^{n}$, where $p_3= P(C(\ba_i)=0)^{\frac{1}{\psi}}$ and $P(C(\ba_i)=0) = e^{-\lambda_E \nu(Y(\ba_i))}$.
\end{proof}
\begin{lemma}\label{lem:d}  $P(D(\ba_1)=0, D(\ba_2)=0, \ldots, D(\ba_n)=0) \le (\sqrt{p_1} + {p_2}^{1/4} + {p_3}^{1/4})^n$.
\end{lemma}
\begin{proof} Follows from \cite[Proposition 3]{Dousse2006}, where event $D(\ba)=1$ if $A(\ba)B(\ba)C(\ba)=1$.
\end{proof}
Let $q \bydef (\sqrt{p_1} + {p_2}^{1/4} + {p_3}^{1/4})$.  The next Lemma characterizes an upper bound on
$q$ for which the probability of having a closed circuit in $\bS$
surrounding the origin is less than one.
\begin{lemma}\label{lem:suffcond} If $q < \frac{11-2\sqrt{10}}{27}$, then the probability of having a closed circuit in $\bS'$
surrounding the origin is less than one.
\end{lemma}
\begin{proof} 
For any circuit of length $n$, there are $4$ possible choices of edges for the starting step and thereafter $3$ choices for every step, except for last step which is fixed given the rest of choice of edges since the circuit has to terminate at the starting point. Moreover, for a circuit containing the origin, the maximum possible distinct intersections with the $x$-axis are $n$. 
Thus, the number of possible circuits of length $n$ around the origin is less than or equal to  $4n3^{n-2}$.
From Lemma \ref{lem:d}, we know that the probability of a closed
circuit of length $n$ is upper bounded by $q^n$. Thus,
\begin{eqnarray*}
P(\text{closed circuit around origin}) &\le & \sum_{n=1}^{\infty} 4n3^{n-2} q^n,\\
&=& \frac{4q}{3(1-3q)^2},
\end{eqnarray*}
which is less than $1$ for $q < \frac{11-2\sqrt{10}}{27}$.
\end{proof}
{\bf Proof of Theorem~\ref{thm:sup}}. Following Lemmas
\ref{lem:int2} and \ref{lem:suffcond}, it suffices to show that
$q$ can be made arbitrarily small for an appropriate choice of
parameters. For any eavesdropper density $\lambda_E$, $p_3$ can be
made arbitrarily small by choosing small enough $s$ and $t$.
Depending on the choice of $s$, $p_1$ can be made arbitrarily
small for large enough legitimate node density $\lambda$, and
finally depending on the choice of $\lambda$, choosing small
enough $\gamma$, $p_2$ can be made arbitrarily small. \qed

\section{Proof of Theorem~\ref{thm:inf}}
In this section, we are interested in the sub-critical regime of percolation, i.e. obtaining a lower bound on $\lambda_c$ as a function of $\lambda_E$ for which percolation does not happen.
We consider the case of $\gamma= 0$, where $x_i$ and $x_j$ are connected in the $SSG$ if $$\ell(d_{ij}) < \left(\frac{\ell(x_i,e)}{1+\sum_{x_j\in \Phi, j\ne i}\ell(x_j,e)}\right), \ \forall \ e\in \Phi_E.$$
If we can show that $\lambda_c > \lambda_0$ for $\gamma =0$, then  since $SSG$ with $\gamma>0$ is contained in $SSG$ with $\gamma=0$, we have that for all $\gamma >0$, $\lambda_c > \lambda_0$. So the lower bound $\lambda_0$ for $\lambda_c$ obtained with $\gamma=0$ serves as a universal lower bound on the critical density $\lambda_c$ required for percolation. Let the interference power received at any eavesdropper with respect to signal from $x_k$ is $I^k_e \bydef \sum_{x_j\in \Phi, j\ne k}\ell(x_j,e)$.

For the case of $\gamma=0$, we proceed as follows. We tile $\bbR^2$
into a square lattice $\bM$ with side $M$. Let  $\bM' =
\bM+(\frac{M}{2}, \frac{M}{2})$ be the dual lattice of $\bM$
obtained by translating each edge of $\bM$ by $(\frac{M}{2},
\frac{M}{2})$. For any edge $\bee$ of $\bM$, let $S_1(\bee)$ and
$S_2(\bee)$ be the two adjacent squares to $\bee$. See Fig.
\ref{fig:subcrit} for a pictorial description. Let $T_1(\bee)$ and
$T_2(\bee)$ be the smaller squares of side $m$ contained inside
$S_1(\bee)$ and $S_2(\bee)$, respectively, as shown in Fig.
\ref{fig:subcrit}, with centers identical to that of $\bS_1$ and
$\bS_2$.

\begin{figure} [ht]
\centering
\scalebox{.6}{\input{JAPSubCrit.pstex_t}}
\vspace{0.1in}
\caption{Open edge definition on a square lattice for sub-critical regime.}
\label{fig:subcrit}
\end{figure}
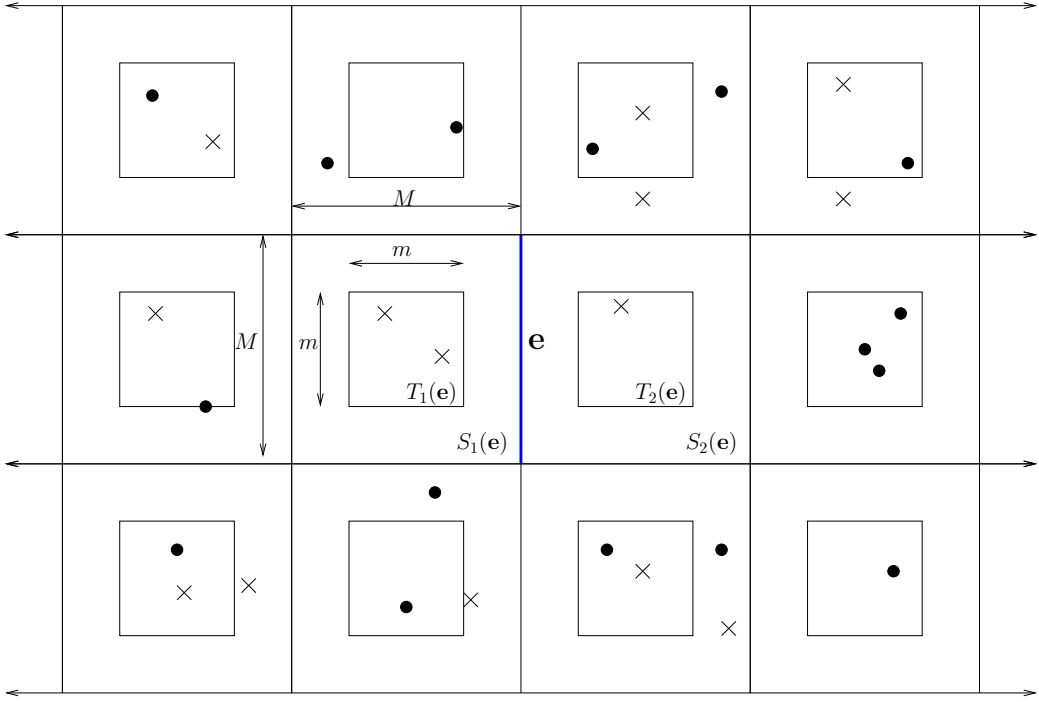

\begin{defn}\label{def:open} For any edge $\bee$ of $\bM$, we define three indicator variables ${\tilde A}(\bee), {\tilde B}(\bee),$ and ${\tilde C}(\bee)$  as follows.\begin{enumerate}
\item ${\tilde A}(\bee) = 1$ if there is at least one eavesdropper node of $\Phi_E$ in both the adjacent squares $T_1(\bee)$ and $T_2(\bee)$.
\item ${\tilde B}(\bee) = 1$ if there are no legitimate nodes in $S_1(\bee)$ and $S_2(\bee)$.
\item ${\tilde C}(\bee) = 1$ if for any eavesdropper node $e \in \Phi_E \cap \left(T_1(\bee)\cup T_2(\bee)\right)$, the interference received from all the  legitimate nodes $I_e \bydef \sum_{x_k\in \Phi}\ell(x_k,e) \le c$.\end{enumerate}
\end{defn}
Then an edge $\bee$ is defined to be open if ${\tilde D}(\bee) = {\tilde A}(\bee){\tilde B}(\bee){\tilde C}(\bee) = 1$.
An open edge is pictorially described in Fig. \ref{fig:subcrit} by a blue edge $\bee$, where the black dots represent legitimate nodes while  crosses are used to
represent eavesdropper nodes.

\begin{lemma}\label{lem:stopping} For any $m$ and $c$, for large enough $M$, an edge $(x_i,x_j)\in SSG$ cannot cross an open edge $\bee$ of $\bM$.
\end{lemma}

\begin{proof} Let two legitimate nodes $x_i, x_j \in\Phi$ be such that the straight line between $x_i$ and $x_j$ intersects an open edge $\bee$ of $\bM$. Then by definition of an open edge, $x_i, x_j \notin  \left(S_1(\bee)\cup S_2(\bee)\right)$. Thus, the signal power between $x_i$ and $x_j$, is $\ell(d_{ij}), d_{ij} > M$. Moreover, the SINR between $x_i$ and any eavesdropper node $e \in \left(T_1(\bee)\cup T_2(\bee)\right)$, $\SINR_{ie} \ge \frac{\ell(d_{ie})}{1+c}$, since edge $\bee$ is open and hence $I_e \le c$ for any $e \in (T_1(\bee)\cup T_2(\bee)$. Thus, choosing $M$ large enough, we can have $\ell(d_{ij}) < \SINR_{ie}$ for any $e \in \left(T_1(\bee)\cup T_2(\bee)\right)$, and hence $x_i$ and $x_j$ cannot be connected directly in $SSG$ if the straight line between them happens to cross an open edge $\bee$ of $\bM$.
\end{proof}

\begin{defn} Consider a square box $B$. Then by $\{L\rightarrow R \ \text{crossing of} \ B \ \text{by}\ G\}$, we mean that there is a connected path of graph $G$ that crosses $B$ from left to right. Similarly, top to bottom crossing is represented as $\{T\rightarrow D \ \text{crossing of} \ B \ \text{by}\ G\}$.
\end{defn}

\begin{lemma}\label{lem:subcritmain} If bond percolation happens on square lattice $\bM$ for large enough $M$ for which an edge $(x_i,x_j)\in SSG$ cannot cross an open edge $\bee$ of $\bM$, then the connected component of $SSG$ is finite.
\end{lemma}
\begin{proof} Consider a square box $B_N$ of side $N$ centered at the origin. Let $M$ be large enough such that an edge $(x_i,x_j)\in SSG$ cannot cross an open edge $\bee$ of $\bM$.
If bond percolation happens on square lattice $\bM$, then
\begin{equation}\label{eq:perceav}\lim_{N\rightarrow \infty} P\left(\exists \ \text{a}\ L\rightarrow R \ \text{and} \ T\rightarrow D \ \text{crossing of} \ B_N \ \text{by open edges of}\ \bM\right) = 1.
\end{equation}
The proof is by contradiction. Let there be an infinite connected component in the $SSG$ with probability $1$. Then, necessarily
\begin{equation}\label{eq:percleg}\lim_{N\rightarrow \infty} P\left(\exists\ \text{a}\ L\rightarrow R\ \text{and} \ T\rightarrow D \ \text{crossing of} \ B_N \ \text{by}\ SSG\right) = 1.
\end{equation}  Since $M$ is such that an edge $(x_i,x_j)\in SSG$ cannot cross an open edge $\bee$ of $\bM$, (\ref{eq:perceav}) and
(\ref{eq:percleg}) cannot hold simultaneously.
\end{proof}

Next, we show that for small enough density of legitimate nodes $\lambda$, bond percolation can happen on a square lattice $\bM$ for large enough $M$ for which an edge $(x_i,x_j)\in SSG$ cannot cross an open edge $\bee$ of $\bM$.

\begin{thm}\label{thm:subcrit} For large enough $M$ that ensures that $(x_i,x_j)\in SSG$ cannot cross an open edge $\bee$ of $\bM$, bond percolation on $\bM$ happens for small enough density of legitimate nodes $\lambda$.
\end{thm}
\begin{proof} Similar to the proof in the super-critical regime, we need to show that the probability of having a closed circuit surrounding the origin in $\bM$ is less than $1$. Towards that end, consider the probability of a closed circuit of length $n$, $P({\tilde D}(\bee_1)=0, {\tilde D}(\bee_2)=0, \ldots, {\tilde D}(\bee_n)=0)$, where ${\tilde D}(\bee_1) = {\tilde A}(\bee_1){\tilde B}(\bee_1){\tilde C}(\bee_1)$.  Similar to  Lemma \ref{lem:a}, $P({\tilde A}(\bee_1)=0, {\tilde A}(\bee_2)=0, \ldots, {\tilde A}(\bee_n)=0)\le   r_1^{n}$, where $r_1 \bydef r_{A}^{1/4}$ and $r_{A} = 1-(1-e^{-\lambda_E m })^2$ is the probability that there is no eavesdropper in either $T_1(\bee)$ or $T_2(\bee)$.
Similarly,  following Lemma \ref{lem:a}, $P({\tilde B}(\bee_1)=0, {\tilde B}(\bee_2)=0, \ldots, {\tilde B}(\bee_n)=0) \le   r_2^{n}$, where $r_2 \bydef r_{B}^{1/4}$ and $r_{B} = 1-e^{-2\lambda M^2}$ is the probability that there is at least one legitimate node of $\Phi$ in $S_1(\bee)$ or $S_2(\bee)$,
$P({\tilde C}(\bee_1)=0, {\tilde C}(\bee_2)=0, \ldots, {\tilde C}(\bee_n)=0) \le r_3^n$ where $r_3  = \bydef e^{\left(\frac{2\lambda}{K}\int \ell(x) dx - \frac{c}{K} \right)}$ following Lemma \ref{lem:b}, and finally $P({\tilde D}(\bee_1)=0, {\tilde D}(\bee_2)=0, \ldots, {\tilde D}(\bee_n)=0) \le \left(\sqrt{r_1} + r_2^{1/4} + r_3^{1/4}\right)^n$ following Lemma \ref{lem:d}. Let $r_s \bydef \sqrt{r_1} + r_2^{1/4} + r_3^{1/4}$.

Using Peierl's argument, bond percolation happens in $\bM$ if $r_s < \epsilon$ for sufficiently small $\epsilon >0$. Let us fix such an $\epsilon> 0$. Then, by choosing $m$ large enough, we can have $\sqrt{r_1} < \frac{\epsilon}{3}$. Moreover, for fixed $c$, let $M$ be large enough such that for any pair of legitimate nodes $x_i, x_j \notin  \left(S_1(\bee)\cup S_2(\bee)\right)$ for which the straight line between them intersects an open edge $\bee$ of $\bM$, $\ell(d_{ij}) < \SINR_{ie}$ for any $e \in \left(T_1(\bee)\cup T_2(\bee)\right)$. Now, given $c,m,$ and $M$, we can  choose $\lambda$ small enough so that $ r_2^{1/4} < \frac{\epsilon}{3}$ and
$ r_2^{1/4} < \frac{\epsilon}{3}$. Thus, we have that $r_s < \epsilon$ as required for an appropriate choice of $m,M$ and $\lambda$.
\end{proof}

Thus, from Theorem \ref{thm:subcrit} and Lemma \ref{lem:subcritmain}, we have that for small enough legitimate node density $\lambda$, $SSG$ with $\gamma=0$ does not percolate.

\bibliographystyle{apt}
\bibliography{../../../Research}

\end{document}

%% file: input.tex
\usepackage{amsfonts}
\usepackage{times}
\usepackage{latexsym}
\usepackage{amssymb}
\usepackage{amsmath}
\usepackage{cite}
\usepackage{verbatim}

\newcommand{\bydef}{\triangleq}

\def\bydef{:=}

\def\bb0{{\mathbb{0}}}

\def\bydef{:=}
\def\ba{{\mathbf{a}}}
\def\bb{{\mathbf{b}}}

\def\bee{{\mathbf{e}}}

\def\bi{{\mathbf{i}}}

\def\b0{{\mathbf{0}}}

\def\bM{{\mathbf{M}}}

\def\bS{{\mathbf{S}}}

\def\b1{{\mathbf{1}}}

\def\bbN{{\mathbb{N}}}

\def\bbR{{\mathbb{R}}}

\def\bydef{:=}

\def\sf0{{\mathsf{0}}}


%% file: JAPSupCrit.pstex_t
\begin{picture}(0,0)%
\epsfig{file=JAPSupCrit.pstex}%
\end{picture}%
\setlength{\unitlength}{3947sp}%
\begingroup\makeatletter\ifx\SetFigFont\undefined%
\gdef\SetFigFont#1#2#3#4#5{%
  \reset@font\fontsize{#1}{#2pt}%
  \fontfamily{#3}\fontseries{#4}\fontshape{#5}%
  \selectfont}%
\fi\endgroup%
\begin{picture}(8949,8424)(664,-9973)
\put(3301,-2011){\makebox(0,0)[lb]{\smash{{\SetFigFont{14}{16.8}{\rmdefault}{\mddefault}{\updefault}{\color[rgb]{0,0,0}$Y(\ba)$}%
}}}}
\put(5176,-5611){\makebox(0,0)[lb]{\smash{{\SetFigFont{14}{16.8}{\rmdefault}{\mddefault}{\updefault}{\color[rgb]{0,0,0}$\bS_2(\ba)$}%
}}}}
\put(4351,-5611){\makebox(0,0)[lb]{\smash{{\SetFigFont{14}{16.8}{\rmdefault}{\mddefault}{\updefault}{\color[rgb]{0,0,0}$\bS_1(\ba)$}%
}}}}
\put(6526,-4111){\makebox(0,0)[lb]{\smash{{\SetFigFont{20}{24.0}{\rmdefault}{\mddefault}{\updefault}{\color[rgb]{0,0,0}$t$}%
}}}}
\put(1305,-5867){\makebox(0,0)[lb]{\smash{{\SetFigFont{17}{20.4}{\rmdefault}{\mddefault}{\updefault}{\color[rgb]{0,0,0}$s$}%
}}}}
\put(4801,-5986){\makebox(0,0)[lb]{\smash{{\SetFigFont{20}{24.0}{\rmdefault}{\mddefault}{\updefault}{\color[rgb]{0,0,0}$\ba$}%
}}}}
\end{picture}%

%% file: JAPSubCrit.pstex_t
\begin{picture}(0,0)%
\epsfig{file=JAPSubCrit.pstex}%
\end{picture}%
\setlength{\unitlength}{3947sp}%
\begingroup\makeatletter\ifx\SetFigFont\undefined%
\gdef\SetFigFont#1#2#3#4#5{%
  \reset@font\fontsize{#1}{#2pt}%
  \fontfamily{#3}\fontseries{#4}\fontshape{#5}%
  \selectfont}%
\fi\endgroup%
\begin{picture}(10824,7284)(-611,-7603)
\put(4876,-3961){\makebox(0,0)[lb]{\smash{{\SetFigFont{20}{24.0}{\rmdefault}{\mddefault}{\updefault}{\color[rgb]{0,0,0}$\mathbf{e}$}%
}}}}
\put(2476,-3961){\makebox(0,0)[lb]{\smash{{\SetFigFont{14}{16.8}{\rmdefault}{\mddefault}{\updefault}{\color[rgb]{0,0,0}$m$}%
}}}}
\put(3451,-2461){\makebox(0,0)[lb]{\smash{{\SetFigFont{14}{16.8}{\rmdefault}{\mddefault}{\updefault}{\color[rgb]{0,0,0}$M$}%
}}}}
\put(3451,-2986){\makebox(0,0)[lb]{\smash{{\SetFigFont{14}{16.8}{\rmdefault}{\mddefault}{\updefault}{\color[rgb]{0,0,0}$m$}%
}}}}
\put(1801,-3961){\makebox(0,0)[lb]{\smash{{\SetFigFont{14}{16.8}{\rmdefault}{\mddefault}{\updefault}{\color[rgb]{0,0,0}$M$}%
}}}}
\put(3601,-4486){\makebox(0,0)[lb]{\smash{{\SetFigFont{14}{16.8}{\rmdefault}{\mddefault}{\updefault}{\color[rgb]{0,0,0}$T_1(\mathbf{e})$}%
}}}}
\put(6001,-4486){\makebox(0,0)[lb]{\smash{{\SetFigFont{14}{16.8}{\rmdefault}{\mddefault}{\updefault}{\color[rgb]{0,0,0}$T_2(\mathbf{e})$}%
}}}}
\put(4126,-5011){\makebox(0,0)[lb]{\smash{{\SetFigFont{14}{16.8}{\rmdefault}{\mddefault}{\updefault}{\color[rgb]{0,0,0}$S_1(\mathbf{e})$}%
}}}}
\put(6526,-5011){\makebox(0,0)[lb]{\smash{{\SetFigFont{14}{16.8}{\rmdefault}{\mddefault}{\updefault}{\color[rgb]{0,0,0}$S_2(\mathbf{e})$}%
}}}}
\end{picture}%